\documentclass{article}

\usepackage{amsmath}
\usepackage{amssymb}
\usepackage{amsthm}

\usepackage{xcolor}
\usepackage{stmaryrd}
\usepackage{fullpage}
\usepackage{wasysym}
\usepackage{pifont}
\usepackage{hyperref}
\hypersetup{colorlinks=true,linkcolor=black,citecolor=black}
\usepackage{cite}
\usepackage{comment}
\usepackage{setspace}

\usepackage[noend]{algorithmic} 
\usepackage{algorithm,caption}

\newcounter{thmcount}

\newtheorem{thm}[thmcount]{Theorem}
\newtheorem{lem}[thmcount]{Lemma}

\definecolor{nts}{rgb}{.8,.1,.8}

\definecolor{hilite}{rgb}{.1,0,.9}
\definecolor{clear}{rgb}{.8,.8,.8}

\definecolor{ozc}{rgb}{.8,.4,.2}
\definecolor{lhc}{rgb}{.4,.2,.6}
\definecolor{lsc}{rgb}{.2,.6,.4}

{\makeatletter
 \gdef\xxxmark{%
   \expandafter\ifx\csname @mpargs\endcsname\relax %
     \expandafter\ifx\csname @captype\endcsname\relax %
       \marginpar{xxx}%
     \else
       xxx %
     \fi
   \else
     xxx %
   \fi}
 \gdef\xxx{\@ifnextchar[\xxx@lab\xxx@nolab}
 \long\gdef\xxx@lab[#1]#2{{\bf [\xxxmark #2 ---{\sc #1}]}}
 \long\gdef\xxx@nolab#1{{\bf [\xxxmark #1]}}
 \long\gdef\xxx@lab[#1]#2{}\long\gdef\xxx@nolab#1{}%
 \gdef\turnoffxxx{\long\gdef\xxx@lab[##1]##2{}\long\gdef\xxx@nolab##1{}}%
}
\newcounter{oztix}
\newcounter{lhtix}
\newcounter{lstix}
\newcommand{\ens}[1]{\ensuremath{#1}}					%
\newcommand{\card}[1]{\ens{|#1|}}							%
\newcommand{\dotlist}[2]{\ens{#1,\ldots,#2}}
\newcommand{\bigoh}[1]{\ens{\mathcal{O}(#1)}}				%

\newcommand{\ith}{\ens{i^{\mbox{\hspace{.2mm}\scriptsize th}}}} 
 
\newcommand{\kth}{\ens{k^{ \mbox{\hspace{.2mm}\scriptsize th}}}}
\newcommand{\zth}{\ens{z^{ \mbox{\hspace{.2mm}\scriptsize th}}}} 
\newcommand{\setbuild}[2]{\ens{\{#1\ |\ #2\}}}				%
\newcommand{\mb}[1]{\mbox{#1}}

\newcommand{\anitem}{\ens{x}}

\newcommand{\valn}{\ens{n}}
\newcommand{\valk}{\ens{k}}
\newcommand{\valnmk}{\ens{\valn - \valk + 1}}

\newcommand{\maxthru}{\textsc{MaxThroughput}}
\newcommand{\mincost}{\textsc{MinCost}}
\newcommand{\kofn}{\valk-of-\valn}
\newcommand{\oneofn}{$1$-of-\valn}

\newcommand{\strategy}{\ens{T}}
\newcommand{\flowtype}{\ens{\tau}}
\newcommand{\thruput}{\ens{F}}
\newcommand{\kfunc}{\ens{f}}                                                      %
\newcommand{\stratspace}{\ens{\mathcal{\strategy}}}		%
\newcommand{\stratspacec}{\ens{\mathcal{\strategy}_c}}	%
\newcommand{\perm}{\ens{\pi}}										%
\newcommand{\loadratio}{\ens{M}}									%
\newcommand{\flowamt}{\ens{m}}									%
\newcommand{\routing}{\ens{R}}										%
\newcommand{\commonr}{\ens{r}}                                                                                                      %
\newcommand{\unsat}{\ens{L}}											%
\newcommand{\satsuff}{\ens{Q}}										%

\newcommand{\processor}{processor}
\newcommand{\Processor}{Processor}
\newcommand{\megaprocessor}{mega\processor}	

\newcommand{\smt}{\textsf{SMT}$(\valk)$ problem}				%
\newcommand{\cmt}{\textsf{CMT}$(\valk)$ problem} 		%

\newcommand{\retval}[1]{\ens{\anitem_{#1}}}					%
\newcommand{\pr}[1]{\ens{p_{#1}}}									%
\newcommand{\qr}[1]{\ens{q_{#1}}}									%
\newcommand{\op}[1]{\ens{O_{#1}}}									%
\newcommand{\megaop}[1]{\ens{E_{#1}}}
\newcommand{\test}[1]{\ens{#1}}										%
\newcommand{\cost}[1]{\ens{c_{#1}}}								%
\newcommand{\rate}[1]{\ens{r_{#1}}}								%
\newcommand{\resrate}[1]{\ens{r'_{#1}}}							%
\newcommand{\landprob}[2]{\ens{g(#1,#2)}}					%
\newcommand{\cstrat}[2]{\ens{\strategy^c_{#1}(#2)}}				%
\newcommand{\sstrat}[2]{\ens{\strategy^s_{#1}(#2)}}				%
\newcommand{\varz}[1]{\ens{z_{#1}}}								%
\newcommand{\vary}[1]{\ens{y_{#1}}}									%

\newcommand{\probgen}[1]{\ens{\mathbf{Pr}\left[#1\right]}}
\newcommand{\varprobgen}[1]{\ens{\mathbf{Pr}\left[#1\right]}}

\newcommand{\hitsum}[2]{\ens{X_{#1,#2}}}					%
\newcommand{\probeq}[3]{\probgen{\hitsum{#1}{#2}~=~#3}}

\newcommand{\probge}[3]{\probgen{\hitsum{#1}{#2}~\geq~#3}}
\newcommand{\varprobge}[3]{\varprobgen{\hitsum{#1}{#2}~\geq~#3}}

\newcommand{\flowin}[2]{\ens{f_{#1}(#2)}}

\newcommand{\xitot}{\ens{\mathcal C}}

\newcommand{\mergeop}[2]{\ens{E_{#2}^{(#1)}}}
\newcommand{\mopfirst}[2]{\ens{b{(#1,#2)}}}
\newcommand{\moplast}[2]{\ens{c{(#1,#2)}}}
\newcommand{\picost}[2]{\ens{\xitot_{#1,#2}}}
\newcommand{\minmega}[1]{\ens{E_{\min}^{(#1)}}}
\newcommand{\mindex}[1]{\ens{h(#1)}}
\newcommand{\themergeop}[1]{\ens{\mergeop{#1}{\mindex{#1}}}}
\newcommand{\themergeoplus}[1]{\ens{\mergeop{#1}{\mindex{#1}+1}}}

\newcommand{\opchargeind}[3]{\ens{\kappa_{#1}(#2,#3)}}
\newcommand{\opchargev}[2]{\ens{\kappa_{#1}(#2)}}
\newcommand{\opcharge}[1]{\ens{\kappa_{#1}}}
\newcommand{\indexsym}{\ens{i'}}
\newcommand{\chargenumv}[2]{\ens{\indexsym(j,#1,#2)}}

\newcommand{\chargenumvshort}[1]{\ens{\indexsym(#1)}}

\title{Max-Throughput for (Conservative) $k$-of-$n$ Testing}
\author{Lisa Hellerstein\thanks{Polytechnic Institute of NYU. This research is supported by the NSF Grant CCF-0917153. {\tt hstein@poly.edu}}
 \and 
\"Ozg\"ur \"Ozkan\thanks{Polytechnic Institute of NYU. This research supported by US Department of Education Grant P200A090157. {\tt ozgurozkan@gmail.com}}
\and
Linda Sellie\thanks{Polytechnic Institute of NYU. This research is supported by a CIFellows Project postdoc, sponsored by NSF and the CRA. {\tt sellie@mac.com}}}

\begin{document}
\maketitle

\begin{abstract}
We define a variant of \kofn{} testing
that we call {\em conservative} \kofn{} testing.
We present a polynomial-time, combinatorial algorithm for the
problem of maximizing throughput of
conservative \kofn{} testing, in
a parallel setting.  
This extends previous work of Kodialam and Condon et al., 
who presented combinatorial
algorithms for parallel pipelined filter ordering,
which is the special case where $\valk=1$ (or $\valk=\valn$)~\cite{conf/ipco/Kodialam01,conf/pods/CondonDHW06,journals/talg/CondonDHW09}.
We also consider the problem of maximizing throughput for 
{\em standard} \kofn{} testing, and show how to obtain a polynomial-time algorithm based on the ellipsoid method using previous techniques.
\end{abstract}

\section{Introduction}

In {\em standard} \kofn{} testing, there are \valn\ binary
tests, that can be applied to an ``item'' \anitem.  
We use \retval i to denote the value of the \ith\ test on \anitem, 
and treat \anitem{} as an
element of $\{0,1\}^{\valn}$.
With probability $\pr i$,
$\retval{i} =  1$,  and with probability $1-\pr i$,
$\retval{i} =  0$.
The tests are independent, and we are given
$\pr 1, \ldots, \pr{\valn}$.
We need to determine whether at least $\valk$ of the \valn\ tests on \anitem{}
have a value of 0, by applying the tests
sequentially to \anitem.
Once we have enough
information to determine whether this is the case, that is, {\em once
we have observed $\valk$ tests with value 0, or \valnmk\ tests with
value 1}, we do not need to perform further tests.\footnote{In an alternative definition of \kofn{} testing, the task is to determine whether at least $\valk$ of the \valn\ tests have a value of 1.  Symmetric results hold for this definition.}

We define {\em conservative} \kofn{} testing the same way,
except that we continue performing tests
until we have either observed $\valk$ tests with value 0,
or have performed all \valn\ tests.  In particular, we do not stop
testing when we have observed \valnmk\ tests with value 1.

There are many applications where \kofn{} testing problems
arise, including 
quality testing, medical diagnosis, and database query optimization.
In quality testing, an item \anitem{} manufactured by a factory is tested for defects.  If it has at least $\valk$ defects, it is discarded.  
In medical diagnosis, the item \anitem{} is a patient;
patients are diagnosed with a particular disease if they fail at least
$k$ out of $n$ special medical tests. 
A database query may ask 
for all tuples \anitem{} satisfying at least $\valk$ of \valn\ given predicates (typically $\valk=1$ or $\valk=\valn$).

For $\valk=1$, standard and conservative \kofn{} testing
are the same.
For $\valk>1$, the conservative variant is relevant
in a setting where, for items failing fewer than $\valk$ tests,
we need to know {\em which} tests they failed. 
For example, in quality testing, we may want to know which tests
were failed by items failing fewer than 
$\valk$ tests (i.e. those not discarded)
in order to repair the associated defects.

Our focus is on the \maxthru\ problem 
for \kofn{} testing. 
Here the objective is to maximize the 
throughput of a system for \kofn{} testing in a parallel setting where each
test is performed by a separate ``{\processor}''.
In this problem, in addition to the probabilities $\pr i$, there
is a {\em rate limit} $\rate i$ associated with the {\processor} that performs
test $\test i$, indicating that the
{\processor} can only perform tests on $\rate i$ items per unit time.

\maxthru\ problems are closely related to 
\mincost\ problems~\cite{conf/pods/LiuPRY08,DBLP:journals/talg/DeshpandeH12}.
In the \mincost\ problem for \kofn{} testing, 
in addition to the probabilities $\pr i$,
there is a cost $\cost i$ associated with performing the \ith\ test.
The goal is to find a testing strategy (i.e.
decision tree) that minimizes the expected cost of testing an individual
item.    There are polynomial-time algorithms for solving the
\mincost\ problem for standard \kofn{} testing~\cite{salloumphd,salloumbreuer,bendov81,journals/tc/ChangSF90}.

Kodialam was the first to study the \maxthru\ \kofn{} 
testing problem, for the special case where $\valk=1$~\cite{conf/ipco/Kodialam01}.
He gave a \bigoh{n^3\log n} algorithm for the problem.  The algorithm is
combinatorial, but its correctness proof relies on
polymatroid theory.  
Later, Condon et~al.~studied the problem, calling it 
``parallel pipelined filter ordering''. 
They gave two \bigoh{n^2} combinatorial algorithms, 
with
direct correctness proofs~\cite{journals/talg/CondonDHW09}.

\paragraph{Our Results.} In this paper, we extend the previous work by giving a polynomial-time
combinatorial algorithm for the \maxthru\ problem for
conservative \kofn{} testing.
Our algorithm can be implemented to run in time 
\bigoh{n^2}, matching the running time of the algorithms of
Condon et al.~for 1-of-n testing.
More specifically, the running time
is \bigoh{\valn(\log \valn + \valk) + o}, where $o$ varies depending on the
output representation used;
the algorithm can be modified to produce different output representations. We discuss output representations below. 

The \maxthru\ problem for
standard \kofn{} testing appears to be fundamentally different
from its conservative variant.
We leave as an open problem the task of developing
a polynomial time {\em combinatorial} algorithm for this problem.
We show that previous techniques can
be used to obtain a polynomial-time
algorithm based on the ellipsoid method. This approach
could also be used to yield an algorithm, based on the
ellipsoid method, for the conservative variant.

\subparagraph{Output Representation}
For the type of representation used by Condon et al.~in achieving their
\bigoh{n^2} bound, $o = \bigoh{n^2}$.  A more explicit representation
has size $o = \bigoh{n^3}$.
We also describe a new, more compact output representation for which
$o = \bigoh{\valn}$. 

In giving running times, we follow Condon et al.~and 
consider only the time taken by the algorithm 
to produce the output representation.
We note, however, that different output representations may
incur different post-processing costs when we want to use them
them to implement the routings.
For example, the compressed representation has $o = \bigoh{n}$, 
but it requires spending \bigoh{n} time in the worst case 
to extract any permutation of \megaprocessor s stored by the 
\megaprocessor\ representation. We can reduce this complexity to 
\bigoh{\log n} using persistent 
search trees~\cite{DBLP:journals/cacm/SarnakT86}.  
In contrast, the explicit $\bigoh{n^3}$ representation gives direct access
to the permutations.
In practice, the choice of the best output representation can vary
depending on the application and the setting.

For ease of presentation, in our pseudocode we use the \megaprocessor\ representation, which is also used by Condon et al.~\cite{journals/talg/CondonDHW09} in their Equalizing Algorithm.

\section{Related Work}

Deshpande and Hellerstein studied the 
\maxthru\ problem for $\valk=1$, when there are precedence constraints
between tests~\cite{DBLP:journals/talg/DeshpandeH12}.
They also showed a close relationship between
the exact \mincost\ and \maxthru\ problems for \kofn{} testing, when $k=1$.
Their results can be generalized to apply to testing of other functions.

Liu et al.~\cite{conf/pods/LiuPRY08} presented a generic, 
LP based method for converting an 
approximation algorithm for a \mincost\ problem, into an approximation
algorithm for a \maxthru\ problem. Their results are not applicable to
this paper, where we consider only exact algorithms.

Polynomial-time algorithms for the
\mincost\ problem for standard \kofn{} testing were given
by Salloum, Breuer, 
Ben-Dov, and Chang et al.~\cite{salloumphd,salloumbreuer,bendov81,journals/tc/ChangSF90,salloumfaster}.

The problem of how to best order a sequence of tests, in
a sequential setting, has been studied in many
different contexts, and in many different models.   
See for example~\cite{conf/pods/LiuPRY08} and~\cite{journals/talg/CondonDHW09}
for a discussion of related work on                                                                                        
the filter-ordering problem (i.e. the \mincost\ problem for $k=1$)
and its variants, and~\cite{unluyurt2004189} for a general survey 
of sequential testing of functions.

\section{Problem Definitions}

A {\em \kofn{} testing strategy} for tests $1, \ldots, \valn$ is a binary decision tree 
\strategy\ that computes the \kofn{} function, 
$\kfunc:\{0,1\}^n \rightarrow \{0,1\}$, where $\kfunc(x) = 1$ if and only if
\anitem{} contains fewer than $\valk$ 0's.
Each node of \strategy\ is labeled by a variable $\retval i$.
The left child of a node labeled with \retval{i} is associated
with $\retval i = 0$ (i.e., failing test $i$), and the right child 
with $\retval i = 1$ (i.e., passing test $i$). 
Each $\anitem \in \{0,1\}^{\valn}$ corresponds to a root-to-leaf path
in the usual way, and the label at the leaf is  
$\kfunc(x)$.

A \kofn{} testing strategy \strategy\ is {\em conservative} if, 
for each root-to-leaf path leading
to a leaf labeled 1, the path contains exactly \valn\ non-leaf nodes,
each labeled with a distinct variable $\retval{i}$.

Given a permutation $\perm$ of the $n$ tests, we define 
\cstrat{\valk}{\perm} to be the conservative strategy described by the following procedure:
{\it Perform the tests
in order of permutation $\perm$ until at least $\valk$ 0's have been observed,
or all tests have been performed, whichever
comes first.  Output $0$ in the first case, and $1$ in the second.}

Similarly, we define \sstrat{\valk}{\perm} to be the following
standard \kofn{} testing strategy:
{\it Perform the tests
in order of permutation $\perm$ until at least $\valk$ 0's have been observed,
or until $n-\valk+1$ 1's have been observed, whichever
comes first.  Output $0$ in the first case, and $1$ in the second.}

Each test $\test i$ has an associated probability $\pr i$, where $0 < \pr i < 1$.
Let $D_p$ denote the product distribution on $\{0,1\}^{\valn}$ 
defined by the $\pr i$'s; that is, if $x$ is drawn from $D_p$, then
$\forall i, \probgen{\retval i=1}=\pr i$ and the \retval{i} are independent.
We use $x \sim D_p$ to denote a random $x$ drawn from $D_p$.
In what follows, when we use an expression of 
the form $\probgen{\ldots}$  involving 
an item $\anitem$, we mean the probability with respect to $D_p$.

\subsection{The \mincost\ problem}
\label{sec:mincostdef}
In the \mincost\ problem for standard \kofn{} testing,
we are given $\valn$ probabilities $\pr i$ and costs $\cost i>0$, for $i \in \{1,\ldots, \valn\}$, associated with the tests.
The goal is to find a \kofn{} testing strategy \strategy\ that minimizes
the expected cost of applying \strategy\ to a random item $\anitem \sim D_p$.
The cost of applying a
testing strategy \strategy\ to an item
\anitem{} is the sum of the costs of the tests along the root-to-leaf 
path for \anitem{} in \strategy. 

In the \mincost\ problem for conservative \kofn{} testing, the goal is the same,
except that we are restricted to finding a {\em conservative} testing strategy.

For example, consider the \mincost\ $2$-of-$3$ problem with probabilities $\pr{1}=\pr{2}=1/2$, $\pr{3}=1/3$ and costs $\cost{1}=1$, $\cost{2}=\cost{3}=2$.
A standard testing strategy for this problem can be
described procedurally as follows: {\em Given item $\anitem$, begin by performing test $1$.  
If $\retval{1} = 1$, 
follow strategy $\sstrat{2}{\pi_1}$, where $\pi_1 = (2,3)$.
Else if $\retval{1}=0$, 
follow strategy $\sstrat{1}{\pi_2}$, where $\pi_2 = (3,2)$.}

Under the above strategy, which can be shown to be optimal,
evaluating $\anitem=(0,0,1)$ costs
$5$, and
evaluating 
$\anitem'=(1,1,0)$ costs $3$.
The expected cost of applying this strategy to a random item $\anitem \sim D_p$ is $3\frac{5}{6}$.

Because the \mincost\ testing strategy may be a 
tree of  size exponential in the
number of tests,  algorithms for the
\mincost\ problem may output a compact representation of the output strategy.
\paragraph{The Algorithm for the \mincost\ Problem.}

In the literature, versions of the \mincost\ problem 
for \oneofn\ testing are studied under a variety of different
names, including pipelined filter ordering, selection ordering, 
and satisficing search (cf.~\cite{journals/talg/CondonDHW09}).

The following is a well-known, simple
algorithm for solving the \mincost\ problem for standard \oneofn\ testing
(see e.g.~\cite{GAREY73}):
First, sort the
tests in increasing order of the ratio $\cost i/(1-\pr i)$.
Next, renumber the tests, so that
$\cost 1/(1-\pr 1) < \cost 2/(1-p_2) < \ldots < \cost{\valn}/(1-\pr{\valn})$.
Finally, output the sorted list $\perm = (1,\ldots, n)$ of tests, which
is a compact representation of the strategy \sstrat{1}{\perm} (which is the same as \cstrat{1}{\perm}).

The above algorithm can be applied to the \mincost\ problem
for conservative \kofn{} testing, simply by treating $\perm$ as
a compact representation of the conservative strategy \cstrat{\valk}{\perm}.
In fact, that strategy is optimal for conservative \kofn{} testing: it has
minimum expected cost among
all conservative strategies.   This follows immediately from a lemma of Boros et al.~\cite{journals/amai/BorosU99}\footnote{The lemma of Boros et al.~actually 
proves that the corresponding decision tree is {\em 0-optimal}.
A decision tree computing a function $f$ 
is 0-optimal if it minimizes the expected cost of
testing an random $x$, {\em given that $f(x) = 0$}.
In conservative \kofn{} testing, where $f$ is the
\kofn{} function, the cost of testing $x$ is the same for all $x$ such that $f(x) = 1$.
Thus the problem of finding a min-cost conservative strategy for \kofn{} testing
is essentially equivalent to the problem of finding a 0-optimal decision tree
computing the \kofn{} function.  The lemma of Boros et al.~also applies to a more
general class of functions $f$ that include the \kofn{} functions.
}.

\subsection{The \maxthru\ problem}

The \maxthru\ problem for \kofn{} testing
is a natural generalization of the 
\maxthru\ problem for
\oneofn\ testing,
first studied by Kodialam~\cite{conf/ipco/Kodialam01}. 
We give basic definitions and motivation here.
For further information about this problem, including
information relevant to its application in practical settings, 
see~\cite{conf/ipco/Kodialam01,conf/pods/CondonDHW06,journals/talg/CondonDHW09}.

In the \maxthru\ problem for \kofn{} testing, as in the \mincost\ problem, 
we are given the probabilities
$\pr 1, \ldots, \pr{\valn}$ associated with the tests.
Instead of costs $\cost i$ for the tests, we are given
{\em rate limits} $\rate i > 0$.
The \maxthru\ problem arises in the following context.
There is an (effectively infinite) stream of items \anitem{} that need to be tested.
Every item \anitem{} must be assigned a strategy \strategy\ that will
determine which tests are performed on it.
Different items may be assigned to different strategies.
Each test is performed by a separate ``{\processor}'', and the {\processor}s operate in parallel.
(Imagine a factory testing setting.)
Item \anitem{} is sent from {\processor} to {\processor} for testing, 
according to its strategy \strategy.
Each {\processor} can only test one item at a time.
We view the problem of assigning items to strategies as
a flow-routing problem.

{\Processor} $\op{i}$ performs test $\test i$.  It has rate limit (capacity) $\rate i$, indicating
that it can only
process $\rate i$ items \anitem{} per unit time. 

The goal is to determine how many items should be assigned to each strategy \strategy,
per unit time, in order to maximize 
the number of items that can be processed per unit time, the throughput of the system.
The solution must respect the rate limits of the {\processor}s, in that the
expected number of items that need to be tested by {\processor} $\op i$ per unit time must
not exceed $\rate i$.  
We assume that tests behave according to expectation: if \flowamt\ items
are tested by {\processor} $\op i$ per unit time, then $\flowamt\pr i$ of them will have the value 1,
and $\flowamt(1-\pr i)$ will have the value 0.

Let \stratspace\ denote the set of all \kofn{} testing strategies and
$\stratspacec$ denote the set of all conservative \kofn{} testing strategies. 
Formally, the \maxthru\ problem for standard \kofn{}
testing is defined by the linear program below.
The linear program defining the \maxthru\ problem for conservative \kofn{} testing is obtained 
by simply replacing the set of \kofn{} testing strategies \stratspace\ by the
set of conservative \kofn{} testing strategies $\stratspacec$.

We refer to a feasible assignment to the variables $\varz{\strategy}$ in the LP below as a {\em routing}.
We call constraints of type (1) {\em rate constraints}.
The value of $\thruput$ is the {\em throughput} of the routing.
 We define $\landprob{\strategy}{i}$ as the probability that test $\test i$ will be performed on an item \anitem{} that is tested using strategy \strategy, when $\anitem \sim D_p$.
For $i \in \{1,\ldots,n\}$, if 
$\sum_{\strategy \in \stratspace} \landprob{\strategy}{i}\varz{\strategy} = \rate i$, 
we say that the routing {\em saturates} {\processor} $\op i$.

We will refer to the \maxthru\ problems for standard and conservative \kofn{} testing as the ``\smt" and the ``\cmt", respectively. 

As a simple example, consider 
the following \cmt{} (equivalently, \smt) instance, 
where $k=1$ and $n=2$: $r_1 = 1$, $r_2 = 2$, $p_1 = 1/2$, $p_2 = 1/4$.
There are only two possible strategies, $T_1(\pi_1)$, where $\pi_1 = (1,2)$,
and $T_1(\pi_2)$, where $\pi_2 = (2,1)$.  
Since all flow assigned to $T_1(\pi_1)$ is tested by $\op{1}$, $\landprob{T_1(\pi_1)}{1}=1$; 
this flow continues on to $\op{2}$ only if it passes test 1, which happens with probability $\pr{1} = 1/2$,
so $\landprob{T_1(\pi_1)}{2}=1/2$.  
Similarly, $\landprob{T_1(\pi_2)}{2}=1$ while $\landprob{T_1(\pi_2)}{1}=1/4,$ since $\pr{2} = 1/4$.
Consider the routing that assigns
$\thruput_1=4/7$ units of flow to strategy $T_1(\pi_1)$, and $\thruput_2=12/7$ units 
to strategy $T_1(\pi_2)$.  Then the amount of flow reaching $\op{1}$ is
$4/7\cdot\landprob{T_1(\pi_1)}{1}+ 12/7\cdot\landprob{T_1(\pi_2)}{1}=1$, 
and the amount of flow reaching $\op{2}$ is
$4/7\cdot\landprob{T_1(\pi_1)}{2}+ 12/7\cdot\landprob{T_1(\pi_2)}{2}=2$.  Since $\rate{1} = 1$ and $\rate{2} = 2$,
 this routing saturates both {\processor}s.   By the results of Condon et al.~\cite{journals/talg/CondonDHW09}, it is optimal.

\vspace{16pt}
\hrule
\vspace{6pt}

{\bf \noindent \maxthru\ LP:}

Given $\rate 1, \ldots, \rate{\valn} > 0$ and $\pr 1 \ldots, \pr{\valn} \in (0,1)$,
find an assignment to the variables $\varz{\strategy}$, for all $\strategy \in \stratspace$,
that maximizes
\[
\thruput = \sum_{\strategy \in \stratspace} \varz{\strategy}
\] 
subject to the constraints:\\[3pt] 
\mbox{\ \ \ \ }(1) $\sum_{\strategy \in \stratspace} \landprob{\strategy}{i}\varz{\strategy} \leq \rate i \mbox{ for all }i \in \{1, \ldots, \valn\}$
\mbox{ and }\\[3pt]
\mbox{\ \ \ \ }(2) $\varz{\strategy} \geq 0 \mbox{ for all }\strategy \in \stratspace$\\[3pt]
where 
$\landprob{\strategy}{i}$ denotes the probability that  test $\test i$ will be performed on an item \anitem{} that is tested using strategy \strategy, when $\anitem \sim D_p$.

\vspace{6pt}
\hrule
\vspace{10pt}

\section{The Algorithm for the \cmt}

We begin with some useful lemmas.
The algorithms of Condon et al.~\cite{journals/talg/CondonDHW09} for maximizing throughput
of \oneofn\ testing rely crucially on the fact that saturation
of all {\processor}s implies optimality.  
We show that the same holds for conservative \kofn{} testing.

\begin{lem}
Let \routing\ be a routing for an instance of the \cmt.  If \routing\ saturates all {\processor}s, then it is optimal.
\end{lem}

\begin{proof}
Each {\processor} $\op i$ can test at most $\rate i$ items per unit time.
Thus at {\processor} $\op i$, there are at most $\rate i (1-\pr i)$ tests 
performed that have the value 0.
Let $\kfunc$ denote the \kofn{} function.

Suppose \routing\ is a routing achieving throughput $\thruput$.
Since $\thruput$ items enter the system per unit time, 
$\thruput$ items must also leave the system per unit time.
An item \anitem{} such that $\kfunc(x) = 0$ does not leave the system until it fails $\valk$ tests.
An item \anitem{} such that $\kfunc(x) = 1$ does not leave the system until it has had all
tests performed on it.
Thus, per unit time,
in the entire system,
the number of tests performed that have the
value 0 must be $\thruput \cdot \loadratio$, where $\loadratio = (k\cdot \probgen{\mbox{\anitem{} has at least $\valk$ 0's}} + \sum_{j = 0}^{k-1}j\cdot \probgen{\mbox{\anitem{} has exactly $j$ 0's}})$. 

Since at most $\rate i (1-\pr i)$ tests 
with the value 0 can occur per unit time at {\processor} $\op i$,
$\thruput \cdot \loadratio \leq \sum_{i=1}^n \rate i (1-\pr i)$.
Solving for $\thruput$, this gives an upper bound of
$\thruput \leq \sum_{i=1}^n \rate i (1-\pr i)/ \loadratio$
on the maximum throughput.  This bound is tight if
all {\processor}s are saturated, and hence a routing saturating
all {\processor}s achieves the maximum throughput.
\end{proof}

In the above proof, we rely on the fact that every routing with throughput $\thruput$
results in the same number of 0 test values being generated in the system per unit time.
Note that this is not the case for {\em standard} testing, where the number of 0 test
values generated can depend on the routing itself,
and not just on the throughput of that routing.
We now give a simple counterexample showing that, in fact, saturation does not  imply optimality for the \smt.
Consider the \maxthru\ $2$-of-$3$ testing instance where
$\pr{1}=1/2,\pr{2}=1/4,\pr{3}=3/4$, and $\rate{1}=2,\rate{2}=1\frac{3}{4},\rate{3}=1\frac{3}{4}$.  

The following is a $2$-of-$3$ testing strategy: 
{\em Given item $\anitem$, peform test 1. If $\retval{1}=1$, 
follow strategy $\sstrat{1}{\pi_1}$, where $\pi_1 = (2,3)$.
Else if $\retval{1}=0$, 
follow strategy $\sstrat{1}{\pi_2}$, where $\pi_1 = (3,2)$.}

Assigning 2 units of flow to this strategy saturates the {\processor}s: 
$\op{1}$ is saturated since it receives the $2$ units entering the system, $\op{2}$ is saturated since it receives 
$1=2\cdot\pr{1}$ units from $\op{1}$ and $3/4=2\cdot\pr{3}\cdot(1-\pr{1})$ items from $\op{1},\op{2}$.  
Similarly, $\op{3}$ is saturated since it receives $1=2\cdot(1-\pr{1})$ units from $\op{1}$ and 
$3/4=2\cdot(1-\pr{3})\cdot\pr{1}$ units from $\op{1}\op{3}$.

We show that the routing is not optimal by giving a different routing with higher throughput.
The routing uses two strategies.
The first is as follows: {\em Given item $\anitem$, perform test 1.
If $\retval{1}=1$, follow strategy $\sstrat{2}{\perm_1}$, where $\perm_1 = (3,2)$.
Else, if $\retval{1}=0$ 
follow strategy $\sstrat{1}{\perm_1}$, where $\perm_2 = (2,3)$.}
The second strategy used by
the routing is $\sstrat{2}{\perm_3}$, where
$\perm_3 = (3,2,1)$.
Assigning $\thruput=1\frac{1}{2}$ units to the first strategy uses
$1\frac{1}{2}$ units of the capacity of $\op{1}$,
$15/16=1\frac{1}{2}\cdot(1-\pr{1})+1\frac{1}{2}\cdot\pr{1}\cdot(1-\pr{3})$ units of the capacity of $\op{2}$, and $15/16=1\frac{1}{2}\cdot(1-\pr{1})+1\frac{1}{2}\cdot(1-\pr{1})\cdot\pr{2}$ of the capacity of $\op{3}$.
This leaves $\op{2}$ and $\op{3}$ with residual capacity more than
$3/4< 1+3/4-15/16$, and $\op{1}$ with residual capacity $1/2=2-1\frac{1}{2}$.
We can then assign $3/4$ additional units to the second strategy
without violating any of the rate constraints, for a routing with total throughput $2\frac{1}{4}$.
(The resulting routing is not optimal, but illustrates our point.)

The routing produced by our algorithm for the \cmt{}
uses only strategies of the form \cstrat{\valk}{\perm}, for some permutation $\perm$ of the tests
(in terms of the LP, this means $\varz{\strategy} > 0$ only if $\strategy = \cstrat{k}{\perm}$ for some $\perm$).
We call such a routing a {\em permutation routing}. We say that it has a {\em saturated suffix}
if for some subset $\satsuff$ of the {\processor}s (1)
\routing\ saturates all {\processor}s in $\satsuff$, and 
(2) for every strategy $\cstrat{\valk}{\perm}$ used by \routing, the {\processor}s in $\satsuff$ (in some order) must form a suffix of $\perm$.

With this definition, and the above lemma, we are now able to generalize a key lemma of Condon et al.~to apply to conservative \kofn{} testing.  The proof is essentially the same as theirs; we present it 
below for completeness.

\begin{lem}
\label{saturatedsuffix}
(\mbox{Saturated Suffix Lemma})  Let \routing\ be a permutation routing for an instance of the \cmt.
If \routing\ has a saturated suffix, then \routing\ is optimal.
\end{lem}

\begin{proof}
If \routing\ saturates all {\processor}s, then the previous lemma guarantees its optimality.
If not, let \unsat{} denote the set of {\processor}s not saturated by \routing.
Imagine that we removed the rate constraints for each {\processor} in \unsat.  
Let $\routing'$ be
an optimal routing for the resulting problem.
We may assume that on any input \anitem, $\routing'$ performs the tests in \unsat{} in some fixed arbitrary order
(until and unless $\valk$ tests with value 0 are obtained), 
prior to performing any tests in $\satsuff$.
This assumption is without loss of generality, because 
if not, we could modify $\routing'$ to first perform the tests in \unsat{}
without violating feasibility, since the {\processor}s in \unsat{} have no rate constraints,
and performing their tests first can only decrease the load on the other {\processor}s.
Thus the throughput attained by $\routing'$ is $T_R \cdot \frac{1}{p_L}$,
where $T_R$ denotes the maximum throughput achievable just with the {\processor}s
in $\satsuff$, and $p_L$ is the probability that a random \anitem{} will have the value 0
for fewer than $\valk$ of the tests in \unsat{} (i.e. it will not be eliminated by the
tests in \unsat).

Routing \routing\ also routes flow first through \unsat, and then through $\satsuff$. 
Since it saturates the {\processor}s in $\satsuff$, by the previous lemma, 
it achieves maximum possible throughput with those {\processor}s.
It follows that \routing\ achieves the same throughput as $\routing'$, and hence is optimal
for the modified instance where {\processor}s in \unsat{}
have no rate constraints.  Since removing constraints can only increase the maximum
possible throughput, it follows that \routing\ is also optimal for the original
instance.
\end{proof}

\subsection{The Equal Rates Case}

We begin by considering the \cmt{} in the special case where the rate limits $\rate i$ are equal to some constant
value \commonr\ for all {\processor}s.
Condon et al.~presented a closed-form solution for this case 
when $\valk = 1$~\cite{journals/talg/CondonDHW09}.  
The solution is a permutation routing that uses \valn\ strategies of the form $T_1(\perm)$.
Each permutation $\perm$ is one of the \valn\ 
left cyclic shifts of the permutation $(1, \ldots, n)$.
More specifically,
for $i \in \{1, \ldots, n\}$, let $\perm_i  = (i, i+1, \ldots, n, 1, 2, \ldots, i-1)$,
and let $T_i = \cstrat{1}{\perm_i}$.
The solution assigns $\commonr(1-\pr{i-1})/(1-\pr 1 \cdots \pr{\valn})$ units of flow to each $T_i$ (where $\pr{0}$ is defined to be $\pr{\valn}$).
By simple algebra,
Condon et al.
verified that the solution saturates all {\processor}s.
Hence it is optimal.

The solution of Condon et al.~is based on the fact that for the
\oneofn\ problem, assigning $(1-\pr{i-1})$ flow
to each $T_i$ equalizes the load on the {\processor}s.
Surprisingly, this same assignment equalizes the load for the \kofn{} problem as well.
Using this fact, we obtain a closed-form solution to the
\cmt.

\begin{lem}
\label{lem:equalrates}
Consider an instance of the {\cmt}. 
For $i \in \{1, \ldots, n\}$, let $T_i$ be as defined above.
Let $\hitsum ab = \sum_{\ell=a}^b (1-\retval \ell)$
and let $\alpha = \sum_{t=1}^k \probge 1\valn t$. 
Any routing that assigns a total of $t$
units of flow to
the strategies $T_i$, such that the fraction of the total that is assigned to each $T_i$ is
$(1 - \pr{i-1})/\sum_{j=1}^n (1 - \pr{j-1})$,
will cause each processor's residual capacity to be reduced by
$t \alpha/\sum_{j=1}^n (1 - \pr{j})$ units.
If all processors have the same rate limit $\commonr$, 
then the routing that assigns
$\commonr(1-\pr{i-1})/\alpha$ units of flow to strategy $T_i$
saturates all {\processor}s.
\end{lem}

\begin{proof}
We begin by considering the routing in which $(1-\pr{i-1})$ units of flow
are assigned to each $T_i$.
Consider the question of how much flow arrives per unit time at \processor\ \op 1, under
this routing.
For simplicity, assume now that $\valk = 2$.
Thus as soon as an item has failed 2 tests, it is discarded.
Let  $\qr i = (1-\pr i)$.

Of the $\qr{\valn}$ units assigned to strategy $T_1$, all  $\qr{\valn}$ arrive at \processor\ \op 1.
Of the $\qr{n-1}$ units assigned to strategy $T_{n}$, all $\qr{n-1}$ arrive at \processor\ \op 1,
since they can fail either 0 or 1 test (namely test \valn) beforehand.

Of the $\qr{n-2}$ units assigned to strategy $T_{n-1}$, the number
reaching \processor\ \op 1 is $\qr{n-2}\beta_{n-1}$, where $\beta_{n-1}$ is the probability
that an item fails either 0 or 1 of tests $\valn-1$ and \valn.
Therefore, $\beta_{n-1} = 1- \qr{n-1}\qr{\valn}$.

More generally, for $i \in \{1, \ldots, n\}$, of the $\qr{i-1}$ units assigned to $T_i$,
the number reaching \processor\ \op 1 is $\qr{i-1}\beta_i$,
where $\beta_i$ is the probability that a random item fails a total of 0 or 1 of
tests $i, i+1, \ldots, n$.
Thus, $\beta_i = \probeq i\valn 0 +\probeq i\valn 1$.
It follows that the total flow arriving at \processor\ \op 1 is 
$\sum_{i=1}^{n} (\qr{i-1} \probeq i\valn 0)  + \sum_{i=1}^n (\qr{i-1}\probeq i\valn 1).$

Consider the second summation,
$\sum_{i=1}^n (\qr{i-1}\probeq i\valn 1)$.
We claim that this summation is equal to $\probge 1\valn 2$, which
is the probability that \anitem{} 
has {\em at least} two \retval{i}'s that are 0.
To see this,
consider a process where we observe the value of $\retval{\valn}$, 
then the value of $\retval{n-1}$ 
and so on down towards $\retval 1$, 
stopping if and when we have observed exactly two 0's.
The probability that we will stop at some point, having observed two 0's, is clearly
equal to the probability that \anitem{} has {\em at least} two \retval{i}'s that are set to 0.
The condition
$\sum_{j=i}^{n} (1 - \retval j) = 1$ is satisfied when
exactly 1 of $\retval{n}, \retval{n-1}, \ldots, \retval i$ has the value 0.
Thus $\qr{i-1}\probeq i\valn 1$
is the probability that we observe exactly one $0$ in 
$\retval{\valn}, \ldots, \retval i$, and then
we observe a second 0 at $\retval{i-1}$.  That is, it is the probability
that we stop after observing $\retval{i-1}$.
Since the second summation takes the sum of 
$\qr{i-1}\probeq i\valn 1$ over all $i$ between 1 and \valn,
the summation is precisely equal to 
the probability of stopping at some point in the above process, having
seen two 0's.  This proves the claim.

An analogous argument shows that the first summation,
$\sum_{i=1}^n (\qr{i-1}\probeq i \valn 0)$,
is equal to 
$\probge 1\valn 1$.

It follows that the amount of flow reaching \processor\ \op 1 is
$\probge 1\valn 1 + \probge 1\valn 2$. 
This expression is symmetric in the {\processor} numbers, so the
amount of flow reaching every $\op i$ is equal to this value.
Thus the above routing causes all {\processor}s to receive the same amount of flow.

Scaling each assignment in the above routing by a constant
factor scales the amount of flow reaching each
{\processor} by the same factor.
In the above routing, the fraction of total flow assigned to each
$T_i$ is 
$\qr{i-1}/\sum_{j=1}^n \qr{j}$, so
each unit of input flow sent along the $T_i$ 
results in each processor receiving
$(\probge 1\valn 1 + \probge 1\valn 2)
/\sum_{j=1}^n \qr{j}$ units.
Thus any routing that assigns a total of $t$
units of flow to the strategies $T_i$, such
that the fraction assigned to each $T_i$ is 
$\qr{i-1}/\sum_{j=1}^n \qr{j}$,
will cause each processor to receive 
$t(\probge 1\valn 1 + \probge 1\valn 2)
/\sum_{j=1}^n \qr{j}$ units.

Thus if
all {\processor}s
have the same rate limit \commonr,  
the routing that assigns $\commonr\qr{i-1}/(\probge 1\valn 1 +  \probge 1\valn 2)$
units to each strategy $T_i$ will saturate all {\processor}s. 

The above argument for $\valk = 2$ can easily be extended to arbitrary $\valk$.
The corresponding
proportional distribution of flow for arbitrary $\valk$ assigns
a $\qr{i-1}/\sum_{j=1}^n \qr{j}$ fraction of the total flow
to strategy $T_i$, and
each unit of input flow sent along the $T_i$ according to
these proportions results in $\alpha/\sum_{j=1}^n \qr{j}$ units
reaching each {\processor}.
The saturating routing
for arbitrary $\valk$, when all {\processor}s have rate limit \commonr, assigns
$r\qr{i-1}/\alpha$ units of flow
to strategy $T_i$.
 \end{proof}

\subsection{The Equalizing Algorithm of Condon et al.}

Our algorithm for the \cmt{} is an adaptation of
one of the two \maxthru\ algorithms, for the special case where $\valk=1$,
given by Condon et al.~\cite{journals/talg/CondonDHW09}.
We begin by reviewing
that algorithm, which we will call the {\em Equalizing Algorithm}.
Note that when $\valk=1$, it only makes sense to consider strategies
that are permutation routings, since an item can be discarded as
soon as it fails a single test.

Consider the \cmt{} for $\valk=1$.  View the problem
as one of constructing a flow of items through the
{\processor}s.  The capacity of each {\processor} is
its rate limit,
and the amount of flow sent along a permutation  $\perm$ (i.e., assigned to
strategy $\cstrat{1}{\perm}$)
is equal to the number of items sent along that path per unit time.
Sort the tests by their rate limits, and re-number them so that
$\rate{\valn} \geq \rate{n-1} \geq \ldots \geq \rate 1$.
Assume for the moment that all rate limits $\rate i$ are distinct.

The Equalizing Algorithm constructs a flow
incrementally as follows.   Imagine pushing 
flow along the single permutation 
$(\valn, \ldots, 1)$.
Suppose we continuously increase the amount of flow being pushed,
beginning from zero,
while monitoring the
``residual capacity'' of each
{\processor}, i.e., the difference between its rate
limit and the amount of flow it is already receiving.
(For the moment, do not worry about
exceeding the rate limit of a {\processor}.) 

Consider two adjacent {\processor}s, $i$ and $i-1$.
As we increase the amount of flow, the residual capacity of
each decreases continuously.  
Initially, at zero flow,
the residual capacity of $i$ is greater
than the residual capacity of $i-1$.  
It follows by continuity that the
residual capacity of $i$ cannot become less than
the residual capacity of $i-1$ without the two residual
capacities first becoming equal.
We now impose the following stopping condition: 
increase the flow sent along permutation $(\valn, \ldots, 1)$
until either
(1) some {\processor} becomes saturated,
or (2) the residual capacities of at least two of the {\processor}s
become equal.  
The second stopping condition ensures that when the flow increase is halted,
permutation $(\valn, \ldots, 1)$ still orders the {\processor}s in
decreasing order of their residual capacities.
(Algorithmically, we do not increase the
flow continuously, but instead directly calculate the amount of flow
which triggers the stopping condition.)

If stopping condition (1) above holds
when the flow increase is stopped, 
then the routing can be shown to have a saturated suffix, and hence
it is optimal.

If stopping condition (2) holds, we keep the current flow, and then
augment it by solving a new \maxthru\ problem in which
we set the rate limits of the {\processor}s to be equal to their residual
capacities under the current flow (their $\pr i$'s remain the same).

We solve the new \maxthru\ problem as follows.
We group the {\processor}s into equivalence
classes according to their rate limits.
We then replace each equivalence class with a single
{\megaprocessor}, with a rate limit equal to the
residual capacities of the constituent {\processor}s, and probability $\pr i$ equal
to the product of their probabilities.
We then essentially apply the procedure for
the case of distinct rate limits to the {\megaprocessor}s.
gen
The one twist is the way in which we
translate flow sent through a {\megaprocessor}
into flow sent through the constituent {\processor}s of
that {\megaprocessor};
we route the flow through
the constituent {\processor}s so as to equalize their load.
We accomplish this by dividing the flow
proportionally between the cyclic shifts of a permutation of the {\processor}s,
using the proportional allocation of Lemma~\ref{lem:equalrates}.
We thus ensure that the {\processor}s in each equivalence class
continue to have equal residual capacity.  
Note that, under this scheme, the residual capacity
of a {\processor} in a {\megaprocessor} may decrease more slowly
than it would if all flow were sent directly to that {\processor}
(because some flow may first be filtered through other {\processor}s
in the {\megaprocessor})
and this needs to be taken into account in determining when
the stopping condition is reached.

We illustrate the
Equalizing Algorithm on the following
\cmt{} where $k=1$ and $n=3$
(since $k=1$ this is also an \smt{}, where $k=1$ and $n=3$).
Suppose we have 3 {\processor}s, $\op{1},\op{2},\op{3}$ 
with rate limits $\rate 1 = 3, \rate 2 = 14$, and $\rate 3 = 18$, and
probabilities $\pr 1 = 1/8,\pr{2} = 1/2$ and $\pr{3} = 1/3$.
When flow is sent along  $\op{3},\op{2},\op{1}$, after 6 units of flow is sent we achieve a stopping condition with $\op{3}$ and $\op{2}$ having the same residual capacity of $12$; the residual capacity of $\op{1}$ is $2$.

Our algorithm then performs a recursive call where the {\processor}s $\op{3}$ and $\op{2}$ are combined into a {\megaprocessor} $\op{2,3}$ with 
associated probability
$\pr{2,3}=1/2\cdot 1/3=1/6$.  Within 
{\megaprocessor} $\op{2,3}$, flow will be routed by sending $3/7$ of it along
permuatation $\op{3},\op{2}$, and the remaining $4/7$ along
permutation $\op{2},\op{3}$; we observe that for one unit of flow sent through $\op{2,3}$ the amount of capacity used by each {\processor} is $3/7 + 2/7=5/7$.
Using this internal routing for {\megaprocessor} $\op{2,3}$,
the algorithm sends flow along $\op{2,3},\op{1}$; after 12 units of flow, we reach a stopping condition when $\op{1}$ is saturated.  Even though $\op{2}$ and $\op{3}$ are not saturated (they have $12-12\cdot 5/7$ residual capacity left) the flows constructed as described provide optimal throughput.

The Equalizing Algorithm, implemented in a straightforward way, 
outputs a representation of the resulting routing that consists
of a sequence of pairs of the form 
$((\dotlist{\megaop m}{\megaop 1}),\hat{t})$,
one for each
recursive call.    We call this a {\em {\megaprocessor} representation}.
The list 
$(\dotlist{\megaop m}{\megaop 1})$
represents the permutation
of \megaprocessor s $\megaop i$ along which flow is sent during that call.
Each $\megaop i$ is given by the subset of original processors contained in it,
and $\hat{t} > 0$ is a real number that
denotes the amount of flow to be sent along 
$(\dotlist{\megaop m}{\megaop 1})$.
Of course, flow coming into each {\megaprocessor} should be routed so as to
equalize the load on each of its constituent {\processor}s.
The size of this representation is $\bigoh{\valn^2}$.
Interpreted in a straightforward way, the representation corresponds
to a routing
that sends flow along an exponential number of different
permutations of the original {\processor}s.

Condon et al.~describe a combinatorial method to reduce the number of 
such permutations used to
be $\bigoh{n^2}$~\cite{journals/talg/CondonDHW09}.  After such a reduction,
the output can be represented explicitly as a set of $\bigoh{n^2}$ 
pairs of the form
$(\pi, t)$, one for each permutation $\pi$ that is used,
indicating that $t > 0$ amount of flow should be sent along permutation
$\pi$.   We call such a representation a {\em permutation  representation}.
The size of this permutataion representation, given explicitly, is $\bigoh{n^3}$. (Hellerstein and Deshpande
describe a linear algebraic method for reducing the number of permutations to be
at most $n$, yielding an explicit reprsentation of size $\bigoh{n^2}$, but at the
cost of higher time complexity\cite{DBLP:journals/talg/DeshpandeH12}.)

We also describe a variant of the \megaprocessor\ representation called the \textit{compressed representation}, where the algorithm outputs only the first permutation explicitly, and the outputs the sequence of merges, yielding a representation of size \bigoh n.

\subsection{An Equalizing Algorithm for the \cmt}

In this section, we prove the following Theorem by presenting an algorithm. We will give an outline of the algorithm as well as its pseudocode. We will then describe how to achieve the running time stated in the Theorem. 

\begin{thm}
\label{thm:alg}
There is a
combinatorial algorithm for solving the {\cmt} that can be implemented to run in time
$\bigoh{n(\log n + k) + o}$, where the value of $o$ depends on the output representation.
For the \megaprocessor\ representation, $o = \bigoh{n^2}$, for the permutation representation,
$o = \bigoh{n^3}$, and for the compressed representation, $o=\bigoh{n}$. 
\end{thm}

\paragraph{Algorithm Outline}
We extend the
Equalizing Algorithm of Condon et al., to apply to 
arbitrary values of $\valk$.
Again, we will push flow along the permutation of the {\processor}s
$(\valn, \ldots, 1)$ (where $\rate{\valn} \geq \rate{n-1} \geq \ldots \geq \rate 1$)
until one of the two stopping conditions is reached: 
(1) a {\processor} is saturated, or
(2) two {\processor}s
have equal residual capacity.
Here, however, we do not discard an item until it has failed $\valk$
tests, rather than discarding it as soon as it fails one test.
To reflect this, we divide the flow into $\valk$ different types, numbered
0 through $\valk-1$, depending
on how many tests its component items have failed.
Flow entering the system is all of type 0.

When \flowamt\ units of flow of type \flowtype{} enters a {\processor} \op i,
$\pr i\flowamt$ units pass test $i$, and $(1-\pr i)\flowamt$ units
fail it.  So, if $\flowtype < k-1$, then of the \flowamt\ incoming
units of type \flowtype, $(1-\pr i)\flowamt$ units 
will exit {\processor} \op i
as type $\flowtype+1$ flow, and $\pr i\flowamt$ will exit as type \flowtype{}
flow.  Both types will be passed on to the next {\processor} in
the permutation, if any.
If $\flowtype = k-1$, then $\pr i\flowamt$ units will exit as
type \flowtype{} flow and be passed on to
the next {\processor}, and the remaining $(1-\pr i)\flowamt$ will be discarded.

Algorithmically, we need to calculate the minimum amount of
flow that triggers a stopping condition.  This computation is
only slightly more complicated for general $\valk$ than it is for $\valk = 1$.
The key is to compute,
for each {\processor} \op i, what fraction of the flow that is pushed into
the permutation will actually reach {\processor} \op i (i.e. we need to
compute the quantity $g(\cstrat{k}{\perm},i)$ in the LP.)

If stopping condition (2) holds, we keep the current flow, and
augment it by solving a new \maxthru\ problem in which
we set the rate limits of the {\processor}s to be equal to their residual
capacities under the current flow (their $\pr i$'s remain the same).
To solve the new \maxthru\ problem, we 
again group the {\processor}s into equivalence
classes according to their rate limits, and
replace each equivalence class with a single
{\megaprocessor}, with a rate limit equal to the
rate limit of the constituent {\processor}s, and probability $\pr i$ equal
to the product of their probabilities.

We then want to apply the procedure for
the case of distinct rate limits to the {\megaprocessor}s.
To do this, we need to
translate flow sent into a {\megaprocessor}
into flow sent through the constituent {\processor}s of
that {\megaprocessor}, so as to equalize their load.
We do this translation separately for each type of flow
entering the {\megaprocessor}.  Note that flow of type \flowtype{} 
must be discarded as soon as it fails an additional $k-\flowtype$
tests. 
We therefore send flow of type \flowtype{} 
into the constituent {\processor}s of the {\megaprocessor}
according to the proportional allocation of
Lemma~\ref{lem:equalrates} for $(k-\flowtype)$-of-$\valn'$ 
testing, where $\valn'$ is the number of consituent {\processor}s
of the {\megaprocessor}.
We also need to compute how much flow of each type
ends up leaving the {\megaprocessor} (some of the incoming flow of
type \flowtype{} entering
the {\megaprocessor} may, for example, become outgoing flow of type $\flowtype+n'$),
and how much its residual capacity is reduced by the incoming flow.

We give a more detailed description of the necessary computations
in the pseudocode, which we discuss next.  However, the pseudocode 
does not contain all the implementation details, and is not 
optimized for efficiency.  It also gives the 
output using a megaprocessor representation.  Following presentation of the pseudocode,
we discuss how to implement it to
achieve the running times stated in Theorem~\ref{thm:alg} for the different
output representations.

\paragraph{Pseudocode}

The main part of the pseudocode is presented below as Algorithm~\ref{alg:maxthru}. 
The following information will be helpful
in understanding it.

At each stage of the algorithm, the {\processor}s are partitioned into
equivalence classes.  The {\processor}s in each equivalence class constitute
a {\megaprocessor}.  Each equivalence class consists of a contiguous subsequence
of {\processor}s in the sorted sequence $\op n, \ldots, \op 2, \op 1$.
We use $m$ to denote the number of {\megaprocessor}s (equivalence classes).
The {\processor}s in each equivalence class all have the same
residual capacity.
In Step 1 of the algorithm, we partition the {\processor}s into
equivalence classes according
to their rate limits;  two processors are in the same equivalence
class if and only if they have the same rate limit.
We use \megaop i to denote both the \ith\ equivalence class and the \ith\ \megaprocessor. 
In some our examples, we denote a \megaprocessor\ containing \processor s $\{\op i, \op{i+1}, \ldots, \op{j}\}$ by \op{i,i+1,\ldots,j}. 

In Step 2, we compute the amount of flow $\hat{t}$ that 
triggers one of the two stopping conditions.
In order to do this, we need to know the rate at which
the residual capacity of each {\processor} within an equivalence
class $\megaop i$
will be reduced when flow is sent 
down the {\megaprocessor}s in the order $\megaop m, \ldots, \megaop 1$.
We use $\xi(i)$ to denote the amount by which the residual capacity
of the {\processor}s in $\megaop i$ is reduced when
one unit of flow is sent in that order.

The equation for $\xi(i)$ follows from the preceding lemmas and discussion.
We use $f_j(z)$ to denote the amount of flow of type $j$
that would reach {\processor} $z$, if one unit of flow were
sent down the permutation $O_n, \ldots, O_1$, where these
are the original {\processor}s, not the {\megaprocessor}s. 
This is precisely equal to the probability that random item $x$ has exactly
 $j$ 0's in tests $n, \ldots, z+1$.
We compute the value of $f_j(z)$ for all $z$ and $j$ in
a separate initialization routine, given below.
The key here is noticing that if you send one unit of flow
down the {\megaprocessor}s $\megaop m, \ldots, \megaop 1$,
the amount of flow reaching {\megaprocessor} \megaop i is precisely
$f_j(c(i))$, where $c(i)$ is the highest index of a {\processor} in $\megaop i$;
the amount of flow reaching the {\megaprocessor} depends only on how
many 0's have been encountered in test $n, \ldots, c(i)+1$,
and not on the order used to perform those tests.

The quantity $\hat{t}_1$ is the amount of flow sent down
$\megaop m, \ldots, \megaop 1$ that would cause saturation of the {\processor}s in $\megaop 1$.
The quantity $\hat{t}_2$ is the minimum amount of flow sent down
$\megaop m, \ldots, \megaop 1$ that would cause the residual capacities of two
{\megaprocessor}s to equalize.  The stopping condition holds at the
minimum of these two quantities.

\renewcommand{\algorithmiccomment}[1]{// #1}

\renewcommand{\algorithmicrequire}{\textbf{Input:}}
\renewcommand{\algorithmicensure}{\textbf{Output:}}
\begin{algorithm*}
\caption*{\maxthru\ Initialization}
\label{maxthruinit}
\begin{algorithmic}
\STATE $\flowin{j}{z} \leftarrow 0,\ \forall z\in\{\dotlist{1}{\valn}\}, \forall j\in\{\dotlist{0}{k-1}\}$;
\STATE $\flowin{0}{1} \leftarrow 1$;
\FOR{$(z\leftarrow 2; z\leq \valn;z \leftarrow z+1)$}
	\FOR{$(j\leftarrow 0;j\leq k-1;j\leftarrow j+1)$}
		\STATE $\flowin {j}{z} \leftarrow \qr{z-1}\flowin{j-1}{z-1} + \pr{z-1}\flowin{j}{z-1}$;
	\ENDFOR
\ENDFOR
\RETURN SolveMaxThroughput(\dotlist{\pr 1}{\pr \valn},\dotlist{\rate 1}{\rate \valn});
\end{algorithmic}
\end{algorithm*}

\begin{algorithm}
\caption{SolveMaxThroughput(\dotlist{\pr 1}{\pr \valn},\dotlist{\rate 1}{\rate \valn})}
\label{alg:maxthru}
\begin{algorithmic}
\REQUIRE \valn\ selectivities \dotlist{\pr 1}{\pr \valn}; \valn\ rate limits $\rate 1 \leq \ldots \leq \rate \valn$
\ENSURE representation of solution to the \maxthru\ problem for the given input parameters
\STATE
\STATE \textbf{1.} \COMMENT{form the equivalence classes \dotlist{\megaop m}{\megaop 1}};
\STATE Let $1\leq \ell_{1}<\ldots<\ell_{m+1}= n+1$ such that, for all $y,y'\in[\ell_{i-1},\ell_{i})$ and $z,z'\in[\ell_{i},\ell_{i+1})$, where $i\in[2,n]$, \;
\STATE \hspace{3mm} we have $\rate y = \rate{y'}<\rate z = \rate{z'}$\;
\STATE Then, for $i\in[1,m]$, $\megaop i = \setbuild{\op z}{\ell_{i}\leq z < \ell_{i+1}}$, and $R_{i} \leftarrow \rate{\ell_{i}}$. 
\STATE
\STATE \textbf{2.} // calculate $\hat{t}$ using the following steps;
\FOR {$(i\leftarrow 1; i\leq m; i\leftarrow i+1)$}
\STATE $c(i) \leftarrow \mbox{highest index of a {\processor} in \megaop i}$;
\STATE $b(i) \leftarrow \mbox{lowest index of a {\processor} in \megaop i}$;
\STATE Recall that $\hitsum ab = \sum_{\ell=a}^b (1-\retval \ell)$\;
\STATE $\xi(i) \leftarrow  \sum_{j=0}^{k-1} \flowin{j}{c(i)}\cdot\left(	\sum_{v= 1}^{k-j} 	\varprobge{b(i)}{c(i)}{v}		\right)/\sum_{t =b(i)}^{c(i)}(1-\pr t)$;
\ENDFOR
\STATE $\hat{t}_1 \leftarrow \frac{R_1}{\xi(1)}$;
\STATE $\hat{t}_2 \leftarrow \min_{i\in [2,\ldots,m]} \left(\frac{R_i - R_{i-1}}{\xi(i) - \xi(i-1)}\right)$;
\STATE $\hat{t} \leftarrow \min(\hat{t}_1,\hat{t}_2)$;
\STATE
\STATE \textbf{3.} \COMMENT{calculate the residual capacity for each {\processor} \op \ell}; 
\FOR{$(\ell\leftarrow 1;\ell\leq \valn; \ell\leftarrow \ell+1)$}
\STATE $j\leftarrow \mbox{ index of the equivalence class \megaop j containing {\processor} \op \ell}$;
\STATE $\resrate \ell \leftarrow \rate \ell - \xi(j)\hat{t}$;
\ENDFOR
\STATE
\STATE \textbf{4.} \COMMENT{store new flow and recurse if needed}
\STATE $K\leftarrow((\dotlist{\megaop m}{\megaop 1}),\hat{t})$;
\IF[residual capacity of equivalence class \megaop 1 is 0]{$(\resrate 1 == 0)$}
\RETURN $K$;
\ELSE
\STATE $K' \leftarrow \mbox{ SolveMaxThroughput(\dotlist{\pr 1}{\pr \valn},\dotlist{\resrate 1}{\resrate \valn})}$;
\RETURN $K \circ K'$; \COMMENT{i.e.~the concatenation of $K$ and $K'$}
\ENDIF
\end{algorithmic}
\end{algorithm}

\paragraph{Example}
We illustrate our algorithm for the \cmt{} on the following example.
Let $k = 2$ and $n = 4$.
Suppose the probabilities are $\pr{1}=\pr{2}=\pr{3}=1/2,$ $\pr{4}=3/4$, and 
the rate limits are $\rate{1}=\rate{2}=12$, $\rate{3}=\rate{4}=10$.  

Our algorithm first combines {\processor}s with same rate limits
into {\megaprocessor}s; thus we combine $\op{1}$ and $\op{2}$ into {\megaprocessor} $\op{1,2}$ with rate limit 12.
It routes flow through this {\megaprocessor} 
by sending a $1/2$ fraction of the flow in the order $\op{1}, \op{2}$, 
and sending the other $1/2$ fraction in the order $\op{2},\op{1}$.  
Similarly, $\op{3}$ and $\op{4}$ have the same rate limit, so they are combined into a {\megaprocessor} 
$\op{3,4}$ with rate limit
10, where a $1/3$ fraction of the flow is sent along $\op{3},\op{4}$, and the
other $2/3$ fraction is
sent along $\op{4},\op{3}$. 

Our {\megaprocessor} $\op{1,2}$ has a higher rate limit than $\op{3,4}$, consequently our 
algorithm routes flow in the order $\op{1,2},\op{3,4}$.  We now show that the stopping condition is reached after sending 
$6$ units of flow along this route.

The $6$ units of flow decreased the capacity of {\processor}s
$\op{1}$, and $\op{2}$ in $\op{1,2}$ by 6, since $k=2$ and thus flow cannot be discarded before it has been subject
to at least two tests.  

We now calculate the reduction of
capacity in $\op{3}$ and $\op{4}$ caused by the $6$ units of flow sent through $\op{1,2},\op{3,4}$.
Flow leaving $\op{1,2}$ 
 has a $1/4$ probability of having failed both {\processor}s in 
$\op{1,2}$ and exiting the system; for flow that stays in the system 
to be tested by $\op{3,4}$, it has a $1/4$ chance of having passed the test of 
both {\processor}s; it has a $1/2$ chance of having passed the test of one {\processor} and having failed 
the test of the other {\processor}.   
Thus, of the 6 units of flow sent into $\op{1,2}$, $1/4\cdot6 = 3/2$ units are 
passed on to $\op{3,4}$ as type 0 flow,
and $1/2\cdot6 = 3$ units of flow are passed on to $\op{3,4}$ as type 1 flow.

Of the $3/2$ units of type 0 flow, 
entering $\op{3,4}$,
all of it must undergo both test 3 and test 4, since flow is not discarded
until it has failed two tests.  Thus that flow reduces the 
capacity of both $\op{3}$ and $\op{4}$ by $3/2$ units.

Of the 3 units of type 1 flow entering $\op{3,4}$, 
$1/3$ is tested first by $\op{3}$, and then by $\op{4}$ only if it passes test $3$ (which it does
with probability $1/2$).
The remaining $2/3$ is tested first by $\op{4}$, and then by $\op{3}$ only if it passes test $4$ (which
is does with probability $3/4$).
Thus of the 3 units of type 1 flow, $3\cdot(1/3+2/3\cdot3/4) = 5/2$ units reach $\op{3}$,
and $3\cdot(2/3+1/3\cdot1/2) = 5/2$ units reach $\op{4}$. 
Hence the $3+3/2$ total units of flow entering $\op{3,4}$ reduce
the capacities of both $\op{3}$ and $\op{4}$ by
$5/2 + 3/2 = 4$.

We have thus shown that the
6 units of flow sent first to $\op{1,2}$ and then to $\op{3,4}$, 
cause the residual capacities of $\op{1}$ and $\op{2}$
to be $12 - 6 = 6$, and the residual capacities of $\op{3}$ and $\op{4}$
to be $10 - 4 = 6$.  Thus the residual capacities of all {\processor}s equalize, as claimed.

At this point our algorithm constructs a new {\megaprocessor}, by combining the {\processor}s in
$\op{1,2}$ with the {\processor}s in $\op{3,4}$.    All the {\processor}s in the resulting
{\megaprocessor}, $\op{1,2,3,4}$, have
a residual capacity of $6$.  Using the proportional allocation of Lemma~\ref{lem:equalrates} to route flow sent into $\op{1,2,3,4}$, we assign
1/7 of the flow into $\op{1,2,3,4}$ to permutation
$\perm_1=\{1,2,3,4\}$, 2/7 to permutation
$\perm_2=\{2,3,4,1\}$, 2/7 to permutation
$\perm_3=\{3,4,1,2\}$, and 2/7 to permutation
$\perm_4=\{4,1,2,3\}$.
By sending a total of 7 units of flow through
$\op{1,2,3,4}$ according to
this allocation, we send 1, 2, 2, and 2 units respectively along
the four permutations, 
achieving the saturating routing
given in Lemma~\ref{lem:equalrates}.

Our final routing achieves a throughput of $6 + 7 = 13$ which is optimal.

\paragraph{Achieving the running time.}

Let us first consider the running time of the algorithm excluding the computation of $\xi(i)$ and the time it takes to construct the output representation $K$. 
It is easy to see that the algorithm makes at most $\valn -1$ recursive calls, because {\megaprocessor}s can only be merged a total of $\valn -1$ times. 
Excluding the computation of $\xi(i)$, the time spent in each recursive call is clearly $\bigoh{n}$. 
However, we can implement the algorithm so as to ensure this time is \bigoh{\log \valn}, as follows. 
First, the maintenance of the equivalence classes can be handled in \bigoh{1} time per merge by simply taking a union of the sets of adjacent {\processor}s in each {\megaprocessor}, instead of recomputing these sets from scratch. 

Second, we do not need to compute the residual capacity of each {\megaprocessor} at every recursive call.
In fact, for all {\megaprocessor}s except the first one, we
only need enough information about its residual capacity to allow us to
compute
$\hat{t}_2 \leftarrow \min_{i\in [2,\ldots,m]} \left(\frac{R_i - R_{i-1}}{\xi(i) - \xi(i-1)}\right)$.
This suggests that for each 
{\megaprocessor} $i$ where $i \geq 2$, we keep the quantity $Q_i$,
where $Q_i = \left(\frac{R_i - R_{i-1}}{\xi(i) - \xi(i-1)}\right)$ 
instead of $R_i$.  
The {\megaprocessor}s can be stored in a priority queue,
according to their $Q_i$ values.

Consider any $i$ where $\megaop i$ or $\megaop{i-1}$ are not involved in a merge. Then 
\begin{align*}
\frac{(R_{i} - \xi(i)\hat t) - (R_{i-1} - \xi(i-1)\hat t)}{\xi(i) - \xi(i-1)} &= 
\frac{R_{i} - R_{i-1}}{\xi(i) - \xi(i-1)} - \hat t. 
\end{align*}
Thus following the merge, $Q_i$ is decreased by the same amount $\hat t$ for
all such $i$.  Therefore, instead of updating the $Q_i$ for these $i$
in the priority queue,
we can keep their current values, and maintain the sum of the $\hat t$ values computed so far; this can be subtracted
from $Q_i$ if its updated value is needed.  
We do need
to remove the two merged {\megaprocessor}s from the priority queue,
insert the information about the resulting new {\megaprocessor},
and update the $Q_i$ values for {\megaprocessor}s $i$ such that $\megaop i$ 
or $\megaop {i-1}$ were involved in a merge. Note that we need to change the $Q_{i}$ values for such \megaprocessor s due to the change in the $\xi(\cdot)$ value of the newly formed \megaprocessor. 
The above operations can be performed in time
\bigoh{\log n} time per merge, using the priority queue.

Therefore, the running time of the algorithm excluding the computation of $\xi(i)$ is \bigoh{n\log n + o} where $o$ is the time required to construct the output. 
In the pseudocode, the output is computed using the \megaprocessor\ representation.
Since there are at most $n$
recursive calls, there are at most $n$ pairs
$((\dotlist{\megaop m}{\megaop 1}),\hat{t})$ in
the output, and therefore $o = \bigoh{n^{2}}$.

If one chose to convert this representation to a permutation representation,
using the combinatorial method of Condon et al.~\cite{journals/talg/CondonDHW09},
then the value of $o$ would be $\bigoh{n^3}$.

Consider instead the following
more compact output representation, which we call the compressed representation. 
Suppose the algorithm outputs the initial permutation, then 
outputs the sequence of merges performed, together with the $\hat{t}$ values
associated with the merges. In this case, we have $o=\bigoh n$.

We will next show that the computation of $\xi(i)$ throughout the algorithm can be performed in \bigoh{\valn\valk} total time. 

\subparagraph{Computing $\xi(i)$.}
Let \mergeop ik be the \kth\ {\megaprocessor} in the
recursive call associated with the \ith\ merge. 
Let \mopfirst ik be the lowest index of a {\processor} in \mergeop ik, and let \moplast ik be the highest index of a {\processor} in \mergeop ik. 
Let $\card{\mergeop ik}$ denote the {\em size} of that {\megaprocessor},
that is, the number of {\processor}s in it.
Thus $\card{\mergeop ik} = \moplast ik - \mopfirst ik + 1$.

Let \mindex i and $\mindex {i}+1$ be the indices of the {\megaprocessor}s merged by the \ith\ merge (i.e.~\themergeop i and \themergeoplus i are the {\megaprocessor}s merged by the \ith\ merge).

Observe that at iteration $i$ after {\megaprocessor} \themergeop i is merged with \themergeoplus i we only recompute $\xi(\mindex i)$. After the merge, we need to compute
\[
\xi(\mindex i) =  \sum_{j=0}^{k-1} \flowin{j}{\moplast i{\mindex i + 1}}\cdot\left(	\sum_{v= 1}^{k-j} \probge{	\mopfirst i{\mindex i}}{\moplast i{\mindex i + 1}}{v}		\right) / 
\sum_{t=\mopfirst i{\mindex i}}^{\moplast i{\mindex i+1}} (1-\pr t)
\]

Consider the denominator 
$\sum_{t=\mopfirst i{\mindex i}}^{\moplast i{\mindex i+1}} (1-\pr t)$
in the above expression.
It is the sum of the failure probabilities of all
{\processor}s contained in 
~\themergeop i and 
\themergeoplus i.
To enable this computation to be performed in constant time per recursive call,
we simply store, with each
{\megaprocessor}, the sum of the failure probabilities of
all {\processor}s in it.  In each recursive call, it only
takes constant time to update this information.
Recall that \flowin{j}{\moplast i{\mindex i+1}} for all $j\in\{0,\ldots,k-1\}$ are computed in the initialization procedure.
Let 
\[
D_{j} = \sum_{v= 1}^{k-j}\probge{\mopfirst i{\mindex i}}{\moplast i{\mindex i + 1}}v	.
\]
Given $D_{j}$ for $j\in\{0,\ldots,k-1\}$, we can compute $\xi(\mindex i)$ in \bigoh k time.
Observe that 
\[
D_{j} = D_{j+1} + \probge{\mopfirst i{\mindex i}}{\moplast i{\mindex i + 1}}{k-j}.
\]
Therefore, given \probge{\mopfirst i{\mindex i}}{\moplast i{\mindex i + 1}}v for $v\in\{1,\ldots,k\}$, we can compute  $\{D_{0},\ldots,D_{k-1}\}$ in \bigoh k time. Finally, observe that
\[
\probge{\mopfirst i{\mindex i}}{\moplast i{\mindex i + 1}}v = \probge{\mopfirst i{\mindex i}}{\moplast i{\mindex i + 1}}{v+1} + \probeq{\mopfirst i{\mindex i}}{\moplast i{\mindex i + 1}}v.
\]
Therefore, given $\probeq{\mopfirst i{\mindex i}}{\moplast i{\mindex i + 1}}v$ for $v\in\{1,\ldots,k\}$, we can compute $\probge{\mopfirst i{\mindex i}}{\moplast i{\mindex i + 1}}v$ for all $v\in\{1,\ldots,k\}$ in \bigoh k time. 
To enable these computations,
we store, with each
{\megaprocessor}, 
the values \probeq{b}{c}v for all $v\in\{1,\ldots,k\}$, 
where $b$ and $c$ are respectively the lowest and highest indices of
the {\processor}s contained in that {\megaprocessor}.
We will analyze below the total cost of keeping these values updated.

We denote by \xitot\ the total cost of computing $\xi(\cdot)$ throughout the algorithm, using the implementation described.
Since we will have to compute $\xi(\cdot)$ at most $n$ times throughout the algorithm, by the arguments above, \xitot\ is bounded by $\bigoh{nk}$ plus the cost of computing $\probeq{\mopfirst i{\mindex i}}{\moplast i{\mindex i + 1}}v$ for all $i\in\{1,\ldots,n-1\}$ and $v\in\{0,\ldots,k\}$. 
Let us denote the cost of computing $\probeq{\mopfirst i{\mindex i}}{\moplast i{\mindex i + 1}}v$ by \picost iv. Therefore, $\xitot = \bigoh{nk} + \sum_{i=1}^{n-1}\sum_{v=0}^{k}\picost iv$. We show that $\xitot = \bigoh{nk}$ by proving $\sum_{i=1}^{n-1}\sum_{v=0}^{k}\picost iv =\bigoh{nk}$. 

\begin{lem}
\label{lem:boundxi}
$\sum_{i=1}^{n-1}\sum_{v=0}^{k}\picost iv = \bigoh{nk}$.
\end{lem}

\begin{proof} Let $\minmega i = \themergeop i$ if $\card{\themergeop i} \leq \card{\themergeoplus i}$ and $\minmega i = \themergeoplus i$ otherwise. 
Recall that when we need to compute $\probeq{\mopfirst i{\mindex i}}{\moplast i{\mindex i + 1}}v$, we have already computed and stored $\probeq{\mopfirst i{\mindex i}}{\moplast i{\mindex i}}v$ and $\probeq{\mopfirst i{\mindex i + 1}}{\moplast i{\mindex i + 1}}v$ for all $v\in\{0,\ldots,k\}$. 

Since $\probeq{b}{c}v = 0$ for any $b,c$ if $v>c-b+1$, we can compute $\probeq{\mopfirst i{\mindex i}}{\moplast i{\mindex i + 1}}v$ using the following equality: 
\begin{equation}
\label{eq:picost}
\probeq{\mopfirst i{\mindex i}}{\moplast i{\mindex i + 1}\hspace{-2mm}}{\hspace{-1mm}v} = 
\hspace{-7mm}
\sum_{j=\max(0,v-\card{\themergeoplus i})}^{\min(v,\card{\themergeop i})} 
\hspace{-10mm}
\probeq{\mopfirst i{\mindex i}}{\moplast i{\mindex i}}j \cdot \probeq{\mopfirst i{\mindex i + 1}}{\moplast i{\mindex i + 1}\hspace{-2mm}}{\hspace{-1mm}v-j}
\end{equation}
Thus, we perform at most one multiplication and one addition for each term in Equation~\ref{eq:picost}, yielding
\begin{equation}
\label{eq:costbound}
\picost iv < 2\cdot\min(v+1,\card{\minmega i}+1).
\end{equation}
We can now bound $\sum_{i,v} \picost iv$ as follows. Each time two {\megaprocessor}s \themergeop i and \themergeoplus i merge, we charge the cost of computing \probeq{\mopfirst i{\mindex i}}{\moplast i{\mindex i + 1}}{v}, for all $v\in\{0,\ldots,k\}$ to the {\processor}s in the smaller of the two {\megaprocessor}s, distributing the cost evenly among the {\processor}s in the {\megaprocessor}. 
Thus, we charge $(\sum_{v=0}^{k}\picost iv)/\card{\minmega i}$ to each {\processor} $\op{i} \in \minmega i$.

Let \opchargeind jiv denote how much of the cost of computing 
$\probeq{\mopfirst i{\mindex i}}{\moplast i{\mindex i + 1}}{v}$ we charge to \op j during the \ith\ merge.
 Let \opchargev jv denote how much of the cost of computing 
$\probeq{\mopfirst i{\mindex i}}{\moplast i{\mindex i + 1}}{v}$ for all $i=\{1,\ldots,k-1\}$ we charge to {\processor} \op j.
  Let \opcharge j denote the total amount we charge to {\processor} \op j.
 In other words, 
\begin{equation}
\label{eq:jiv}
\opchargeind jiv = 
\left\{
\begin{array}{cl}
\picost iv/\card{\minmega i} & \mbox{if } \op j\in\minmega i,  \\
0  &    \mbox{otherwise.}
\end{array}
\right.
\end{equation}
\begin{equation}
\opchargev jv = \sum_{i=1}^{n-1} \opchargeind jiv
\end{equation}
\begin{equation}
\opcharge j = \sum_{v=0}^{k} \opchargev jv = \sum_{i=1}^{n-1}\sum_{v=0}^{k} \opchargeind jiv
\end{equation}
Then, we have $\sum_{i,v} \picost iv = \sum_{j=1}^{n} \opcharge j$.  We will bound $\sum_{i,v}\picost iv$ by proving an upper bound on \opcharge j.

Consider any {\processor} \op j. We will show that $\opchargev jv = \bigoh{1}$. 
Let $\chargenumvshort z = \chargenumv vz$ be the index of the merge in which
{\processor} \op j is charged for the cost of computing $\probeq{b}{c}{v}$ for any $b,c$ for the \zth{}
time. 
Formally, let 
\[
\chargenumvshort z = \chargenumv vz =
\left\{
\begin{array}{cl}
\ell  & \mbox{if } \exists \ell\in[1,n-1] \mbox{ s.t.~} \opchargeind j\ell{v}>0 \wedge \card{\{t\ |\ t<\ell, \opchargeind jtv>0\}} = z-1    \\
0  & \mbox{otherwise}     
\end{array}
\right.
\]

By Equation~\ref{eq:costbound}, $\picost iv < 2\card{\minmega i}+2$, which implies $\opchargeind jiv < 4$ by Equation~\ref{eq:jiv}.
Observe that by the definition of \minmega i and \chargenumvshort z, if $\chargenumvshort 2 > 0$, we have 
\begin{equation}
\label{eq:twiceisv}
\card{\minmega{\chargenumvshort 2}} \geq \card{\mergeop{\chargenumvshort 1}{\mindex{\chargenumvshort 1}}} + \card{\mergeop{\chargenumvshort 1}{\mindex{\chargenumvshort 1} + 1}} \geq v.
\end{equation}
Also by the definition of \minmega i, if $\chargenumvshort z > 0$, we have 
\begin{equation}
\label{eq:doubling}
\card{\minmega{\chargenumvshort z}} \geq 2\cdot \card{\minmega{\chargenumvshort {z-1}}}.
\end{equation} 
Combining all these facts, and letting $Z = \max(x: \chargenumvshort x > 0)$, we have
\begin{align*}
\opchargev jv 
&= \sum_{i=1}^{n-1} \opchargeind jiv && \mb{by definition}\\
&\leq 4+ \sum_{i=2}^{n-1} \opchargeind jiv && \opchargeind jiv < 4\\
&= 4 + \sum_{z=2}^{Z} \opchargeind j{\chargenumvshort z}v && \mb{by definition}\\
&= 4+ \sum_{z=2}^{Z} \frac{\picost {\chargenumvshort z}v}{\card{\minmega{\chargenumvshort z}}} && \mb{by Equation~\ref{eq:jiv}}\\
&< 4+ \sum_{z=2}^{Z} \frac{2v+2}{\card{\minmega{\chargenumvshort z}}} && \mb{by Equation~\ref{eq:costbound}}\\
&\leq 4+ \sum_{z=2}^{Z} \frac{2v+2}{2^{z-2}\cdot\card{\minmega{\chargenumvshort 2}}} && \mbox{by Equation~\ref{eq:doubling}}\\
&< 6+ \sum_{z=2}^{Z} \frac{2v}{2^{z-2}\cdot\card{\minmega{\chargenumvshort 2}}} && \card{\minmega{\chargenumvshort 2}}\geq 2\\
&\leq 6+ \sum_{z=2}^{Z} \frac{2}{2^{z-2}} && \mb{by Equation~\ref{eq:twiceisv}}\\
&< 6+ 4 \\
&= 10.
\end{align*}
Thus, we have $\opchargev jv = \bigoh{1}$. This yields $\opcharge j = \sum_{v=0}^{k} \opchargev jv = \bigoh{k}$.
Since $\sum_{j=1}^{n} \opcharge j \leq n\cdot \max_{i\in\{1,\ldots,n\}} \opcharge i$, we have, 
\begin{equation*}
\sum_{i=1}^{n-1}\sum_{v=0}^{k}\picost iv = \sum_{j=1}^{n} \opcharge j \leq n\cdot\bigoh{k} = \bigoh{nk}
\end{equation*}
\end{proof}
Therefore, by Lemma~\ref{lem:boundxi}, we have $\xitot = \bigoh{nk} + \sum_{i=1}^{n-1}\sum_{v=0}^{k}\picost iv = \bigoh{nk}.$
Thus, our algorithm runs in total time \bigoh{n(\log n+ k) + o}.

\section{An Ellipsoid-Based Algorithm for the \smt}
\label{sec:smt}

There is a simple and elegant algorithm that solves the
\mincost\ problem for standard \kofn{} testing, due to 
Salloum, Breuer, and
(independently) Ben-Dov~\cite{salloumphd,salloumbreuer,bendov81}.
It outputs a strategy compactly
represented by two permutations,
one ordering the {\processor}s in increasing order of
the ratio
$\cost i/(1-\pr i)$, and the other  
in increasing order of the ratio
$\cost i/\pr i$.  
Chang et al.~and Salloum and Breuer later gave 
modified versions of this algorithm that
output a less compact,
but more efficiently evaluatable
representation of the same strategy
~\cite{salloumfaster,journals/tc/ChangSF90}.

We now show
how to combine previous techniques 
to obtain a polynomial-time algorithm for the \smt{} based on the ellipsoid method.  The algorithm uses a technique of 
Despande and Hellerstein~\cite{DBLP:journals/talg/DeshpandeH12}.
They showed that, for \oneofn\ testing, an algorithm
solving the \mincost\ problem can be combined with the ellipsoid
method to yield an algorithm for the \maxthru\ problem.
In fact, 
as we see in the proof below, their approach is actually a
generic one, and can be applied to the testing of other functions.

The ellipsoid-based algorithm for \kofn{} testing makes use of the dual of the LP for the \cmt, which
is as follows:

\vspace{6pt}
\hrule
\vspace{6pt}
{\bf \noindent Dual of Max-Throughput LP:}
Given $\rate 1, \ldots, \rate{\valn} > 0$, $\pr 1 \ldots, \pr{\valn} \in (0,1)$,
find an assignment to the variables $\vary{i}$, for all $i \in \{1, \ldots, n\}$, 
minimizing
$$\thruput = \sum_{i=1}^{n} \rate i \vary{i}$$
subject to the constraints: \\[3pt]
\mbox{\ \ \ \ }(1) $\sum_{i=1}^{n} g(\perm, i)\vary{i} \geq 1 \mbox{ for all }T \in 
\stratspacec$,\\[3pt]
\mbox{\ \ \ \ }(2) $\vary{i} \geq 0 \mbox{ for all }i \in \{1, \ldots, n\}$.\\
\vspace{6pt}
\hrule
\vspace{10pt}

\begin{thm}
\label{thm:smt}
There is a polynomial-time algorithm, based on the ellipsoid method, for
solving the \smt.
\end{thm}

\begin{proof} 
The  approach of Deshpande and Hellerstein works as follows.  The input consists of
the $\pr i$ and the $\rate i$, and the goal is to solve the
\maxthru\ LP in time polynomial in \valn.
The number of variables of the \maxthru\ LP is not
polynomial, so the LP cannot be 
solved directly.
Instead, the idea is to solve it by first using 
the ellipsoid method to solve the dual LP.
The ellipsoid method is
run using an algorithm that simulates a separation oracle
for the dual in time polynomial in \valn.
During the running of the ellipsoid method, the violated
constraints returned by the separation oracle are saved
in a set \loadratio.  Each constraint of the dual corresponds
to an ordering \strategy.  When the ellipsoid method terminates,
a modified version of the \maxthru\ LP is generated,
which includes only the variables $\varz{\strategy}$ corresponding to
orderings \strategy\ in \loadratio\ (i.e. the other variables $\varz{\strategy}$ are
set to 0).  This modified version can then be solved 
directly using a polynomial-time LP algorithm.
The resulting solution is an optimal
solution for the original \maxthru\ LP.

The above approach requires 
a polynomial-time algorithm for simulating the separation oracle
for the dual.  
Deshpande and Hellerstein's method for simulating the
separation oracle relies on the following observations.
In the dual LP for the
\maxthru\, \oneofn\ testing problem,
there are $n!$ constraints corresponding to the
$n!$ permutations of the {\processor}s.
The constraint for permutation $\perm$
is
$\sum_{i=1}^n \landprob{\strategy_1(\pi)}{i}\vary{i} \leq 1$.
If one views $y$ as a vector of costs, where the cost of $i$
is $\vary{i}$, then 
$\sum_{i=1}^n \landprob{\strategy}{i}\vary{i}$  is the expected cost of testing
an item \anitem{} using ordering \strategy.
Thus one can determine the ordering \strategy\ that minimizes
$\sum_{i=1}^n \landprob{\strategy}{i}\vary{i}$
by solving the \mincost\ problem with probabilities
$\pr 1, \ldots, \pr{\valn}$ and cost vector $y$.
(Liu et al.'s approximation algorithm for generic
\maxthru\ also relies on this observation~\cite{conf/pods/LiuPRY08}.)

If the \mincost\ ordering \strategy\ has expected cost
less than 1, then the constraint it corresponds to is
violated.  Otherwise, since the right hand side of
each constraint is 1, $y$ obeys all constraints.
Thus simulating the separation oracle for
the dual on input $y$ can be done by
first running the \mincost\ algorithm (with probabilities $\pr i$ and costs
$\vary{i}$) to find a \mincost\ ordering \strategy.  Once \strategy\ is found,
the values of the coefficients $\landprob{\strategy}{i}$ are calculated.
These are used to calculate
$\sum_{i=1}^n \landprob{\strategy}{i}$, the expected cost of \strategy.  If this value
is less than 1, then the constraint
$\sum_{i=1}^n \landprob{\strategy}{i}$ is returned.

To apply the above approach to \maxthru\ for standard \kofn{} testing,
we observe that in the dual LP for this problem,
there is a constraint, 
$\sum_{i=1}^n \landprob{\strategy}{i}\vary{i} \leq 1$,
for every possible strategy $\strategy$,
We can simulate a separation
oracle for the dual on input $y$ by
running a \mincost\ algorithm for standard \kofn{} testing.
We also need to be able to compute the $\landprob{\strategy}{i}$ values for the
strategy output by that algorithm.
The algorithm of 
Chang et al.~for the \mincost\ standard \kofn{} testing problem is suitable for this purpose,
as it can easily be modified to output
the $\landprob{\strategy}{i}$ values associated with its output strategy \strategy~\cite{journals/tc/ChangSF90}. 
\end{proof}

\bibliographystyle{plain}
\bibliography{kofnonly}

\end{document}